\documentclass[conference]{IEEEtran}
\IEEEoverridecommandlockouts
\usepackage{cite}
\usepackage{graphicx}
\usepackage{textcomp}
\usepackage{subfigure}
\usepackage{colortbl}
\usepackage[norelsize,ruled,vlined,linesnumbered]{algorithm2e}
\usepackage{algorithmic}
\usepackage{amsthm}
\usepackage{amsmath,amssymb,amsfonts}
\usepackage{booktabs}
\usepackage{url}
\usepackage{hyperref}
\usepackage{cleveref}
\usepackage{xcolor}

\newtheoremstyle{mystyle}% % Name
    {2.0mm}%	% Space above
    {2.0mm}% 	% Space below
    {\it}%		% Body font
    {0.0mm}%	% Indent amount
    {\scshape}% % Theorem head font
    {.}%	% Punctuation after theorem head
    { }%	% Space after theorem head, ' ', or \newline
    {}%		% Theorem head spec (can be left empty, meaning `normal')
\theoremstyle{mystyle}
\newtheorem{definition}{Definition}
\newtheorem{example}{Example}

\newtheorem{corollary}{Corollary}
\newtheorem{theorem}{Theorem}

\newcommand{\vs}{\vspace{1.5mm}}

\def\BibTeX{{\rm B\kern-.05em{\sc i\kern-.025em b}\kern-.08em
    T\kern-.1667em\lower.7ex\hbox{E}\kern-.125emX}}
\begin{document}

\newcommand{\todo}[1]{{\bfseries\color{red}ToDo: #1}}

\title{Duration-constrained Interval Joins}

\author{
\IEEEauthorblockN{Naoya Ehara}
\IEEEauthorblockA{\textit{Graduate School of Information Science and Technology} \\
\textit{The University of Osaka}\\
Osaka, Japan \\
ehara.naoya@ist.osaka-u.ac.jp}
\and
\IEEEauthorblockN{Daichi Amagata}
\IEEEauthorblockA{\textit{Graduate School of Information Science and Technology} \\
\textit{The University of Osaka}\\
Osaka, Japan \\
amagata.daichi@ist.osaka-u.ac.jp}
}

\maketitle

\begin{abstract}
Many databases, including temporal, uncertain, spatial, and trajectory databases, use interval data, and interval joins are among the most frequently used operators.
Many studies proposed efficient interval join algorithms, but they do not consider the overlap duration.
They return any pairs of intervals, even if they overlap very slightly, e.g., with no essential correlation or relationship.
Subsequent applications may suffer from such interval pairs, as they may be noise or unnecessary for the analysis.
Furthermore, outputting such pairs also increases join time.
To address the above issues, this paper addresses the problem of duration-constrained interval join.
Given two interval collections $R$ and $S$ and an overlap duration constraint $\epsilon$, this problem returns all interval pairs $(r,s)$ such that $r \in R$, $s \in S$, and the overlap duration between $r$ and $s$ is at least $\epsilon$.
A straightforward approach for this problem is to run a state-of-the-art interval join algorithm and then filter qualified interval pairs.
However, this is inefficient, as it generates unnecessary interval pairs and incurs duration computations, which cannot overcome the above efficiency concern.
We propose an efficient algorithm for this problem that removes the above drawback.
Furthermore, we propose two optimization techniques to improve the efficiency of our algorithm.
We conduct extensive experiments on three real-world interval datasets, and the results demonstrate that our algorithm outperforms existing techniques applicable to our problem.
\end{abstract}


\section{Introduction}  \label{sec:introduction}
An interval $r$ is a pair of start and end points, i.e., $r = [r.start, r.end]$.
Many database systems, including temporal databases \cite{snodgrass1986temporal,bohlen2017temporal}, uncertain databases \cite{cheng2003probabilistic,dalvi2007efficient}, spatial databases \cite{georgiadis2023raster,georgiadis2025raster}, and trajectory databases \cite{wang2008flexible}, manage interval data.
Hence, efficient processing of interval data remains an important research topic.

Interval joins are one of the most important and frequently used operations in (temporal-related) databases \cite{gao2005join}.
Given two interval collections $R$ and $S$, an interval join returns all interval pairs $(r,s)$ such that $r \in R$ and $s \in S$ overlap.
\Cref{fig:example} illustrates an example of an interval join between $R = \{r_1, r_2\}$ and $S = \{s_1, s_2, s_3, s_4\}$. 
This query has many applications that require entity linking.
For example, consider employee databases that record the working period of each employee in different departments \cite{bouros2017forward,bouros2021memory,dignos2022leveraging,dignos2014overlap}.
An interval join returns all pairs of employees whose working period overlaps.
Moreover, interval joins are used in IoT environments, such as sensor networks \cite{piatov2016interval} and traffic monitoring in data centers \cite{pilourdault2016distributed}.
It is known, in addition, that an interval join can be used for a spatial join and equi-join on uncertain data \cite{bouros2017forward,bouros2021memory}.

\subsection{Motivation}
Although interval joins are useful in real-world applications, their computational cost is high when $|R| = n$ and $|S| = m$ are large, because it requires $O(nm)$ time in the worst case.
To alleviate this cost, many works \cite{gunadhi1991query,bouros2020band,bouros2018interval,piatov2021cache,dignos2022leveraging,piatov2016interval,zhang2023scalable} devised efficient interval join algorithms.
However, they do not consider the \textit{overlap duration} between two intervals.
Overlap joins return $(r,s)$ even if $r$ and $s$ overlap slightly.
In the above employee database example, two employees who almost never work together are included in the output, e.g., $(r_1, s_1)$ in \Cref{fig:example}.
Pairs of slightly overlapping intervals can be noise, potentially affecting analysis accuracy negatively.

Introducing the overlap duration as a constraint can remove this concern and obtain a more reasonable result.
Bioinformatics \cite{sarmashghi2019note}, medical data tests \cite{cho2020clinical}, and overlap functions \cite{da2020general} put importance on the overlap duration between intervals.
In addition, many mobile applications also benefit from this setting.
For example, GPS points can generate mobile user records of the store stay periods \cite{shirai2024estimating}.
Given a set of these intervals, a set of intervals of the store's advertisements (and/or promotions), and a duration constraint, this store can analyze how each promotion affects users while removing slightly overlapping intervals, which can be noise \cite{amagata2025target}.
Another example is life logs.
Wearable devices can record activity sessions and vital health sensor readings.
Assuming a database that maintains such historical records, joining them while considering a sufficient temporal overlap helps analyze causal relationships in users' health status.
Therefore, we consider \textit{duration-constrained interval joins} in this paper.

\subsection{Challenge}  \label{sec:challenge}
A straightforward approach for the problem of duration-constrained interval join is to run a state-of-the-art interval join algorithm and then pick all interval pairs that satisfy the duration constraint among the join result.
This approach has two drawbacks: (i) it requires an additional cost of computing the overlap duration for each join result, and (ii) the original join cost remains, although it can contain many interval pairs that do not satisfy a given overlap duration constraint.
Another approach is to evaluate the overlap duration, instead of evaluating whether overlapping or not, when running the interval join algorithm.
When there are many pairs of overlapping intervals, this approach needs to compute the overlap duration at least for each pair.
Thus, this approach also incurs a huge computational cost.
These challenges suggest that using a state-of-the-art interval join algorithm is not appropriate for duration-constrained interval joins, and an algorithm optimized for this problem is required.

\begin{figure}
    \centering
    \includegraphics[width=0.99\linewidth]{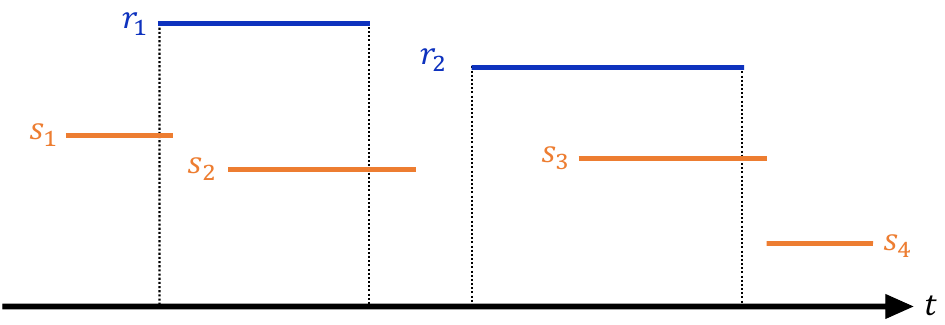}
    \caption{Example of an interval join. The result of join between $R$ and $S$ includes $(r_1,s_1)$, $(r_1, s_2)$, and  $(r_2, s_3)$.}
    \label{fig:example}
\end{figure}

\subsection{Contribution}
To address the above challenges, this paper makes the following contributions.

% \vs
% \noindent
% (1) We address the problem of duration-constrained interval join.
% To our knowledge, this is the first work that solves this problem.

\vs
\noindent
(1) We propose a novel algorithm for the duration-constrained interval join problem.
By exploiting a new thresholding approach with a grid index, our algorithm avoids unnecessary comparisons for both interval pairs included and not included in the result.

\vs
\noindent
(2) We propose two optimization techniques for our algorithm.
The first optimization reduces the cost of comparisons between intervals.
Because our algorithm employs a nested-loop join, our second optimization exploits batch processing for similar intervals.

\vs
\noindent
(3) We conduct extensive experiments on three real-world interval datasets.
The experimental results demonstrate that our algorithm outperforms existing techniques applicable to our problem.

\subsection{Organization}
The rest of this paper is organized as follows.
\Cref{sec:preliminary} introduces the formal problem definition.
\Cref{sec:proposal} presents our algorithm.
We report our experimental results in \Cref{sec:experiment}, and related works are reviewed in \Cref{sec:related-work}.
Finally, this paper is concluded in \Cref{sec:conclusion}.

\section{Preliminary}   \label{sec:preliminary}
An interval $r$ is represented as $r = [r.start, r.end]$, i.e., $r$ starts at $r.start$ and ends at $r.end$.
The interval length of $r$ is denoted by $|r|$, and $|r| = r.end - r.start$.
For ease of understanding, we first define the interval join problem.

\begin{definition}[Interval join problem]   \label{def:join}
Given two interval sets $R$ and $S$, this problem returns $R \bowtie S$ such that
\begin{align*}
    & R \bowtie S = \\
    &\{ (r, s) \,|\, r \in R, s \in S, (r.start \leq s.end) \wedge (s.start \leq r.end) \}.
\end{align*}
\end{definition}

We consider a variant of the above problem.
Let $l(r,s)$ be the overlap duration between $r$ and $s$.
Then,
\begin{equation*}
    l(r,s) = \min\{r.end, s.end\} - \max\{r.start, s.start\}.
\end{equation*}
Note that $\epsilon$ is application-dependent and is specified on demand.
Our problem is defined as follows:

\begin{definition}[Duration-constrained interval join problem]
Given two interval sets $R$ and $S$ and a duration constraint $\epsilon > 0$, this problem returns $R \bowtie_{\epsilon} S$ such that
\begin{equation*}
    R \bowtie_{\epsilon} S = \{ (r, s) \in R \bowtie S \,|\, l(r,s) \geq \epsilon \}.
\end{equation*}
\end{definition}

\noindent
We see that the interval join problem is a special case of our problem with $l(r,s) > 0$.

In this work, as with many existing studies that consider interval data, e.g., \cite{amagata2024independent,amagata2024independent_,amagata2024efficient,amagata2025top,piatov2021cache,bouros2017forward,bouros2018interval,bouros2020band,bouros2021memory,bouros2024hint,piatov2016interval}, we assume that $R$ and $S$ are memory-resident and do not receive (frequent) updates.
The objective of this work is to devise an efficient algorithm for the duration-constrained interval join problem.

\section{Proposed Algorithm}    \label{sec:proposal}
This section presents our proposed algorithm.

\vs
\noindent
\textbf{Main idea.}
A join between $n$ and $m$ objects requires $O(nm)$ time, because the join size is $O(nm)$ in the worst case.
Therefore, to improve the practical efficiency of the duration-constrained interval join, it is important to reduce constant factors, i.e., unnecessary comparisons between intervals.
To this end, we require a technique that identifies whether $(r,s) \in R \bowtie_{\epsilon} S$ or not \textit{without comparisons}.
As the first step for achieving this, we define two thresholds.

\begin{definition}
Given an interval $r \in R$ and a duration constraint $\epsilon$, $\theta_{start}(r)$ and $\theta_{end}(r)$ are respectively defined as follows.
\begin{align*}
    \theta_{start}(r)   &= r.end - \epsilon   \\
    \theta_{end}(r)     &= r.start + \epsilon
\end{align*}
\end{definition}

\noindent
Then, we have the following theorems.
\begin{theorem} \label{theorem:filter}
Given two intervals $r \in R$ and $s \in S$ and a duration constraint $\epsilon$, assume that $|r| \geq \epsilon$ and $|s| \geq \epsilon$.
(i) If $s.start \leq r.start$ and $\theta_{end}(r) \leq s.end$, we have $(r,s) \in R \bowtie_{\epsilon} S$.
(ii) If $\theta_{start}(r) < s.start$ or $s.end < \theta_{end}(r)$, we have $(r,s) \notin R \bowtie_{\epsilon} S$.
\end{theorem}

\begin{proof}
In the case of (i), we have
\begin{align*}
    l(r,s)  & = \min\{r.end, s.end\} - \max\{r.start, s.start\} \\
            & \ge \theta_{end}(r) - r.start     \\
            & = (r.start + \epsilon) - r.start  \\
            & = \epsilon
\end{align*}
The inequality holds since $s.end \geq \theta_{end}(r)$ and $r.end \geq \theta_{end}(r)$.
In the case of (ii), if $\theta_{start}(r) < s.start$, we have
\begin{align*}
    l(r,s)  & < \min\{r.end, s.end\} - \theta_{start}(r)    \\
            & = \min\{r.end, s.end\} - (r.end - \epsilon)   \\
            & \le r.end - (r.end - \epsilon)                \\
            & = \epsilon
\end{align*}
If $s.end < \theta_{end}(r)$, we have
\begin{align*}
    l(r,s)  & < \theta_{end}(r) - \max\{r.start, s.start\}            \\
            & = (r.start + \epsilon) - \max\{r.start, s.start\} \\
            & \le (r.start + \epsilon) - r.start                \\
            & = \epsilon
\end{align*}
Therefore, this theorem holds.
\end{proof}

\begin{example}
\Cref{fig:theorem} illustrates an example of \Cref{theorem:filter}.
We see that $s_1$ and $s_2$ have case (i).
Also, $s_3$ and $s_4$ fall into case (ii).
Only $s_5$ needs to be evaluated whether $l(r,s_5) \geq \epsilon$ or not.
\end{example}

\begin{figure}
    \centering
    \includegraphics[width=0.99\linewidth]{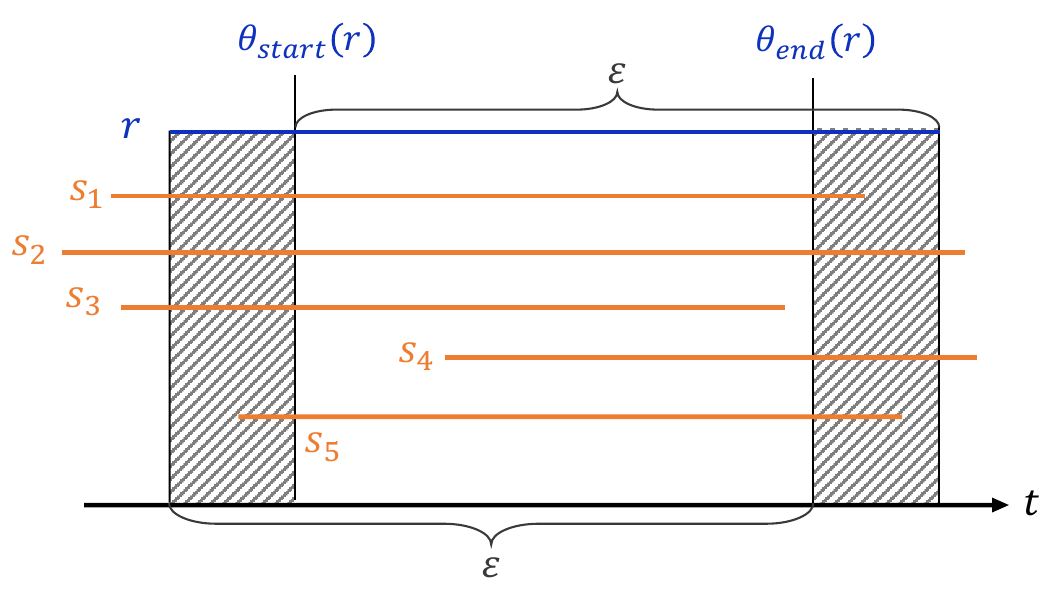}
    \caption{Example of \Cref{theorem:filter}}
    \label{fig:theorem}
\end{figure}

\Cref{theorem:filter} yields several conditions that identify whether $(r,s) \in R \bowtie_{\epsilon} S$ or not, without computing $l(r,s)$.
However, the cost of checking the condition (i) or (ii) is the same as that of computing $l(r,s)$, for given $r$ and $s$.
We thus next need a data structure which, for a given $r \in R$, can find a set $S'$ of intervals $\in S$ falling into case (i) or (ii) \textit{in a batch}.
Specifically, for any $r \in R$, to efficiently find a set of intervals $\in S$ satisfying case (i) or (ii), it is required that $S$ is sorted by start and end points.
In addition, if $(r,s) \in R \bowtie_{\epsilon} S$ and $s$ and $s'$ are similar, $(r,s')$ tends to be in $R \bowtie_{\epsilon} S$.
The data structure, thus, should group similar intervals to avoid redundant interval comparisons.

\begin{figure}
    \centering
    \includegraphics[width=0.65\linewidth]{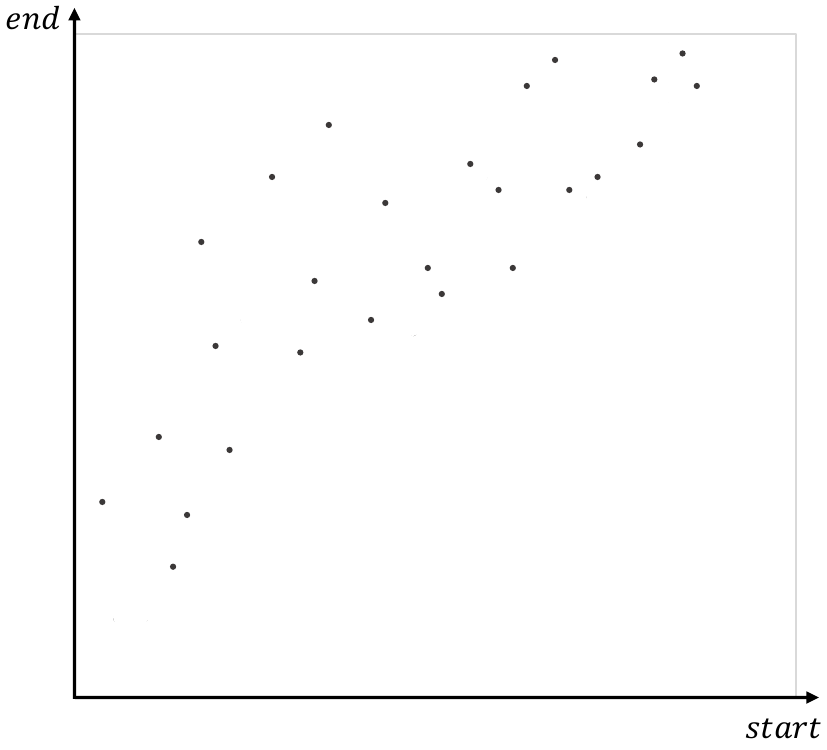}
    \caption{Two-dimensional representation of intervals}
    \label{fig:2d}
\end{figure}

\vs
\noindent
\textbf{Data structure.}
To this end, we employ a two-dimensional grid structure $\mathcal{G}$ to manage $S$.
Note that each interval can be represented as a two-dimensional point; the coordinates of an interval $r$ are $(r.start, r.end)$, and \Cref{fig:2d} illustrates this observation\footnote{Although an interval range query $q$ can be transformed into a 2-sided orthogonal range search $(-\infty, q.end] \times [q.start,\infty)$ on two-dimensional points, this orthogonal range query cannot consider the overlap duration.
Therefore, using a spatial index, such as R-tree and $k$d-tree, is not useful for our problem.}.
Consider $\mathcal{G}$ as a matrix, and each cell of $\mathcal{G}$ is represented as $c_{i,j}$.
We use $S_{i,j}^{start}$ and $S_{i,j}^{end}$ to denote the sets of intervals in $c_{i,j}$ that are respectively sorted based on start and end points (ascending order).
In addition, we maintain
\begin{itemize}
    \setlength{\leftskip}{-2.0mm}
    \item   an array $A_{col}^{max}$ to manage the maximum start point in each column; e.g., $A_{col}^{max}[i]$ is the maximum start point among the intervals in the $i$-th column.
            Notice that $A_{col}^{max}[i] < A_{col}^{max}[i+1]$ for any $i \geq 0$.
\end{itemize}

\noindent
Similarly, for each $i$-th column, we maintain additional arrays:
\begin{itemize}
    \setlength{\leftskip}{-2.0mm}
    \item   $A_{i}^{end,min}$ is an array, where $A_{i}^{end,min}[j]$ maintains the minimum end point among the intervals in $c_{i,j}$.
    \item   $A_{i}^{end,max}$ is an array, where $A_{i}^{end,min}[j]$ maintains the maximum end point among the intervals in $c_{i,j}$.
    \item   $A_{i}^{start,min}$ is an array, where $A_{i}^{start,min}[j]$ is the minimum start point among the intervals in $c_{i,j}$.
\end{itemize}

\vs
\noindent
\textbf{Comparison skipping for $s$ in a batch.}
For a given $r \in R$, these structures support skipping interval comparisons for some intervals $\in S$ \textit{in a batch}, without losing correctness, in the following cases.

\begin{corollary}   \label{corollary:column}
Given an interval $r \in R$,
let $i$ be the minimum column number such that $\theta_{start}(r) \leq A_{col}^{max}[i]$.
Then, for every interval $s$ in $c_{i',j}$ such that $i' > i$ and $j \geq 0$, we have $(r,s) \notin R \bowtie_{\epsilon} S$.
\end{corollary}

\begin{proof}
In this case, we have
\begin{align*}
    \theta_{start}(r) \leq A_{col}^{max}[i] &= \max_{s \,\in\, \cup_{j}\, c_{i,j}}s.start  \\
                        & < \min_{s \,\in\, \cup_{j}\, c_{i',j}}s.start,
\end{align*}
which falls into case (ii) of \Cref{theorem:filter}.
\end{proof}

\begin{corollary}   \label{corollary:cell_1}
Given an interval $r \in R$, if $\theta_{end}(r) \leq A_{i}^{end,min}[j]$ and $s.start \leq r.start$, for every interval $s$ in $c_{i,j}$, we have $(r,s) \in R \bowtie_{\epsilon} S$.
\end{corollary}

\begin{proof}
This case has
$$\theta_{end}(r) \leq A_{i}^{end,min}[j] = \min_{s \,\in\, c_{i,j}}s.end.$$
As $s.start \leq r.start$, this case falls into case (i) of \Cref{theorem:filter}.
\end{proof}

\begin{corollary}   \label{corollary:cell_2}
Given an interval $r \in R$, if $A_{i}^{end,max}[j] < \theta_{end}(r)$, for every interval $s$ in $c_{i,j}$, we have $(r,s) \notin R \bowtie_{\epsilon} S$.
\end{corollary}

\begin{proof}
This case has
$$A_{i}^{end,max}[i] = \max_{s \,\in\, c_{i,j}}s.end < \theta_{end}(r)$$
which falls into case (ii) of \Cref{theorem:filter}.
\end{proof}

\begin{corollary}   \label{corollary:cell_3}
Given an interval $r \in R$, if $\theta_{start}(r) < A_{i}^{start,min}[j]$, for every interval $s$ in $c_{i,j}$, we have $(r,s) \notin R \bowtie_{\epsilon} S$.
\end{corollary}

\begin{proof}
This case has
$$\theta_{start}(r) < A_{i}^{start,min}[j] = \min_{s \,\in\, c_{i,j}} s.start,$$
which falls into case (ii) of \Cref{theorem:filter}.
\end{proof}

\noindent
We exploit these observations to reduce unnecessary interval comparisons.
Note that our algorithm can employ an arbitrary grid structure.
In our experiments, we use the state-of-the-art one \cite{ceccarello2025indexing}.

\vs
\noindent
\textbf{Pre-processing time.}
Assume that $\mathcal{G}$ is built on $S$.
Mapping an interval $\in S$ to the corresponding cell needs $O(1)$ time, so building $\mathcal{G}$ of $S$ takes $O(m)$ time.
Note that we do not create empty cells, thus the number of cells is bounded by $m$.
This means that the size of each array is also bounded by $m$.
Therefore, making the four arrays needs $O(m\log m)$ time.
In total, our pre-processing time is $O(m\log m)$.

\vs
\noindent
\textbf{Space complexity.}
From the above pre-processing time, the space complexity of $\mathcal{G}$ is $O(m)$.
Again, the size of each array is also bounded by $m$.
The space complexity of our data structures is hence $O(m)$.

\vs
\noindent
\textbf{Update.}
Although we do not assume (frequent) updates, our data structures are easy to update.
Given a new interval, we map it to the corresponding cell.
Assume that this cell is $c_{i,j}$.
If this insertion affects $A_{i}^{end,min}[j]$, $A_{i}^{end,max}[j]$, or $A_{i}^{start,min}[j]$, we update the corresponding array.
Furthermore, if this new interval has the maximum start point in the $i$-th column, we update $A_{col}^{max}[i]$.
If an existing interval is removed from $S$, we do similar updates to the above insertion.

\vs
\noindent
\textbf{Algorithm description.}
Our join algorithm employs a nested-loop join approach.
Assume that $|R| \leq |S|$; otherwise, we swap them.
For each $r \in R$, we run \Cref{alg:overlap_query}.
If $|r| < \epsilon$, as it cannot satisfy the overlap duration constraint, we skip this interval.
Otherwise, we set $\theta_{start}(r)$ and $\theta_{end}(r)$.
We then run a binary search on $A_{col}^{max}$ to find the first column number, $idx$, that has $r.start \leq A_{col}^{max}[\cdot]$.
Also, we run a binary search on $A_{col}^{max}$ to find the first column number, $idx'$, that has $\theta_{start}(r) \leq A_{col}^{max}[\cdot]$.
We do not access the $idx' +\alpha$ ($\alpha \geq 1$) column of $\mathcal{G}$ from Corollary \ref{corollary:column}.

Consider each $i$-th column of $\mathcal{G}$ where $i \in [0, idx-1]$.
This column guarantees that $s.start < r.start$ for any $s$ in $\bigcup_j c_{i,j}$ since $A_{col}^{max}[i] < r.start$.
For a cell $c_{i,j}$ of $\mathcal{G}$, if $A_{i}^{end,max}[j] < \theta_{end}(r)$, all intervals in $c_{i,j}$ are not included in $R \bowtie_{\epsilon} S$ from Corollary \ref{corollary:cell_2}.
Hence, we focus on cell $c_{i,j}$ such that $\theta_{end}(r) \leq A_{i}^{end,max}[j]$.
This can be achieved by running a binary search on $A_{i}^{end,max}$.
If $\theta_{end}(r) \leq A_{i}^{end,min}[j]$, from Corollary \ref{corollary:cell_1}, all intervals in $c_{i,j}$ are included in $R \bowtie_{\epsilon} S$, so they are added into the result without comparisons.
If $\theta_{end}(r) > A_{i}^{end,min}[j]$, we add all intervals $s \in c_{i,j}$ such that $\theta_{end}(r) \leq s.end$ without comparisons.

Next, consider each $i$-th column of $\mathcal{G}$ where $i \in [idx, idx' - 1]$.
We run similar operations to the above case.
If $\theta_{start}(r) < A_{i}^{start,min}[j]$, from Corollary \ref{corollary:cell_3}, all intervals in $c_{i,j}$ are not included in $R \bowtie_{\epsilon} S$, thereby we skip this cell.
Otherwise, we add all intervals in $S_{i,j}^{start}$ to the result if $r.end \leq A_{i}^{end,min}[j]$.
In the other cases, we compute $l(r,s)$ for each $s \in S_{i,j}$ and add $(r,s)$ iff $l(r,s)) \geq \epsilon$.

Last, given the $idx'$-th column, we do almost the same operations as above.
One difference is that this column does not have the case of $r.end \leq A_{i}^{end,min}[j]$.

\begin{algorithm}[!t]
    \caption{Proposed Algorithm}    \label{alg:overlap_query}
    \DontPrintSemicolon
    \KwIn{$\mathcal{G}$, $r$, and $\epsilon$}
    \KwOut{$J_{r} = r \bowtie_\epsilon S$}
    \If {$|r| < \epsilon$}{
        \textbf{return}
    }
    $J_{r} \gets \varnothing$\;                     \label{algo:overlap_b}
    $\theta_{start}(r) \gets r.end - \epsilon$\;
    $\theta_{end}(r) \gets r.start + \epsilon$\;
    $idx \gets$ \textsc{Find-Index}$(A_{col}^{max}, r.start)$\;
    $idx' \gets$ \textsc{Find-Index}$(A_{col}^{max}, \theta_{start}(r))$\;
    \For {$i \gets 0$ to $idx - 1$}{
        $j \gets$ \textsc{Find-Index}$(A_{i}^{end,max}, \theta_{end}(r))$\;
        \While{$j \leq |A_{i}^{end,max}|$}{
            $I \gets \varnothing$\;
            \eIf {$\theta_{end}(r) \leq A_{i}^{end,min}[j]$}{
                $I \gets S^{end}_{i,j}$
            }{
                $I \gets \{s \in S^{end}_{i,j} \,|\, \theta_{end}(r) \leq s.end\}$
            }
            \ForEach{$s \in I$}{
                $J_{r} \gets J_{r} \cup \{(r,s)\}$
            }
            $j \gets j + 1$                                                                                             \label{algo:overlap_e}
        }
    }
    \For {$i \gets idx$ to $idx'$}{
        $j \gets$ \textsc{Find-Index}$(A_{i}^{end,max}, \theta_{end}(r))$\;
        \While{$j \leq |A_{i}^{end,max}|$}
        {
            \If {$\theta_{start}(r) < A_{i}^{start,min}[j]$}{
                $j \gets j + 1$\;
                \textbf{continue}
            }
            \eIf{$i < idx'$}{
                \eIf{$r.end \leq A_{i}^{end,min}[j]$}{
                    $J_{r} \gets J_{r} \cup \{(r,s)\}$, $\forall s \in S^{start}_{i,j}$
                }{
                    \ForEach{$s \in S^{start}_{i,j}$ s.t. $l(r,s) \geq \epsilon$}{  \label{algo:comparison_b}
                        $J_{r} \gets J_{r} \cup \{(r,s)\}$                          \label{algo:comparison_e}
                    }
                }
            }{
                Run lines \ref{algo:comparison_b}--\ref{algo:comparison_e}          \label{algo:comparison}
            }
            $j \gets j + 1$
        }
    } 
\end{algorithm}

\begin{example} \label{example:algorithm}
\Cref{fig:algorithm} depicts an example of \Cref{alg:overlap_query} assuming that $r$, such that $|r| \geq \epsilon$, is given.
After identifying $idx$ and $idx'$, this algorithm first accesses the 0-th column.
It runs a binary search on $A_{0}^{end,max}$, and accesses cells in order of $c_{0,3} \rightarrow c_{0,4} \rightarrow c_{0,5}$.
Each yellow cell assumes the case of $\theta_{end}(r) \leq s.end$, so it adds all intervals $s \in c_{0,3}$ such that $\theta_{end}(r) \leq s.end$ without comparisons.
Each orange cell assumes the case of $\theta_{end}(r) \leq A_{i}^{end,min}[j]$, so it adds all intervals $\in c_{0,4} \cup c_{0,5}$ without comparisons.
At the $idx$-th column, the blue cell assumes the case of $r.end \leq A_{i}^{end,min}[j]$.
Hence, all intervals in $S_{idx,3}^{start}$ are added to the result.
The gray cell assumes the case of $\theta_{start}(r) < A_{i}^{start,min}[j]$, so $c_{idx,4}$ is skipped.
At the $idx'$-th column, the green cells assume that the other cases, i.e., it computes $l(r,s)$ for every interval in $c_{i,j}$.
\end{example}

\begin{figure}
    \centering
    \includegraphics[width=0.80\linewidth]{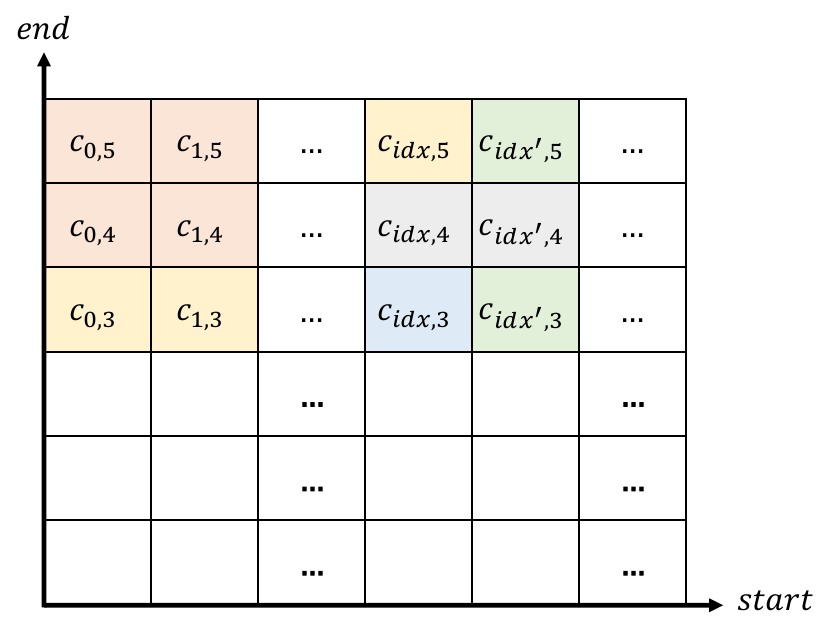}
    \caption{Illustration of Example \ref{example:algorithm}}
    \label{fig:algorithm}
\end{figure}

\subsection{Optimization}   \label{sec:proposal:opt}
Although \Cref{alg:overlap_query} exploits \Cref{theorem:filter} and its corollaries to reduce computing $l(r,s)$ for many intervals $\in S$, at lines \ref{algo:comparison_b} and \ref{algo:comparison}, it computes $l(r,s)$ for each $s \in c_{i,j}$.
We next reduce this computational cost by introducing additional pruning rules.

\begin{definition}
Given a cell $c_{i,j}$ of $\mathcal{G}$ and an overlap duration constraint $\epsilon$, $\theta_{min}(c_{i,j})$ and $\theta_{max}(c_{i,j})$ are respectively defined as follows:
\begin{align*}
    \theta_{min}(c_{i,j})   &= A_{i}^{end,min}[j] - \epsilon   \\
    \theta_{max}(c_{i,j})   &= A_{i}^{end,max}[j] - \epsilon
\end{align*}
\end{definition}

\noindent
This definition yields the following pruning rules:

\begin{theorem} \label{theorem:filter2}
Given an interval $r \in R$ and an overlap duration constraint $\epsilon$, consider $c_{i,j}$ such that $\theta_{end}(r) \leq A_{i}^{end,min}[j]$.
(i) If $s.start \leq r.start$ or $r.start < s.start \leq \theta_{min}(c_{i,j}) \leq \theta_{start}(r)$, we have $(r,s) \in R \bowtie_{\epsilon} S$.
(ii) If $\min \{\theta_{max}(c_{i,j}), \theta_{start}(r)\} < s.start$, we have $(r,s) \notin R \bowtie_{\epsilon} S$.
\end{theorem}

\begin{proof}
Assume that $\theta_{end}(r) \leq A_{i}^{end,min}[j]$ for cell $c_{i,j}$.
In the case of (i), when $s.start \leq r.start$, $l(r,s) \geq \epsilon$ from \Cref{theorem:filter}.
When $r.start < s.start \leq \theta_{min}(c_{i,j}) \leq \theta_{start}(r)$,
\begin{align*}
    l(r,s)  & = \min\{r.end, s.end\} - \max\{r.start, s.start\} \\
            & \ge \min\{r.end, A_{i}^{end,min}[j]\} - \theta_{min}(c_{i,j}) \\
            & = A_{i}^{end,min}[j] - (A_{i}^{end,min}[j] - \epsilon) \\
            & = \epsilon
\end{align*}
Note that $\theta_{min}(c_{i,j}) = A_{i}^{end,min}[j] - \epsilon \leq \theta_{start}(r) = r.end - \epsilon$, so we have $\min\{r.end, A_{i}^{end,min}[j]\} = A_{i}^{end,min}[j]$.
In the case of (ii), when $\theta_{start}(r) < \theta_{max}(c_{i,j})$, $l(r,s) < \epsilon$ from \Cref{theorem:filter}.
When $\theta_{max}(c_{i,j}) \leq \theta_{start}(r) < s.start$, because $\theta_{max}(c_{i,j}) = A_{i}^{end,max}[j] - \epsilon < \theta_{start} = r.end - \epsilon$, we have
\begin{align*}
    l(r,s)  & = \min\{r.end, s.end\} - \max\{r.start, s.start\} \\
            & < A_{i}^{end,max}[j]  - \theta_{max}(c_{i,j})  \\
            & = A_{i}^{end,max}[j] - (A_{i}^{end,max}[j] - \epsilon) \\
            & = \epsilon
\end{align*}
Therefore, this theorem holds.
\end{proof}

\begin{figure}
    \centering
    \includegraphics[width=0.99\linewidth]{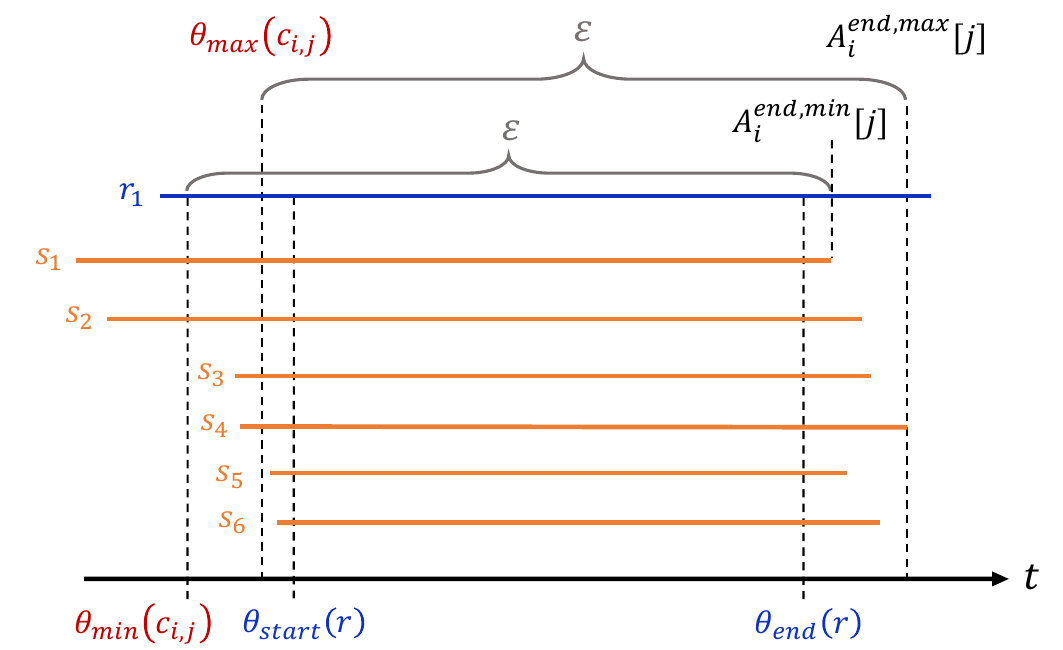}
    \caption{Example of \Cref{theorem:filter2}}
    \label{fig:theorem2}
\end{figure}

\begin{example}
Assume $c_{i,j}$ having lines \ref{algo:comparison_b} and \ref{algo:comparison} of \Cref{alg:overlap_query}.
\Cref{fig:theorem2} describes this case.
We see that $s_1$ and $s_2$ have case (i).
Also, $s_5$ and $s_6$ have case (ii) as they satisfy $\theta_{max}(c_{i,j}) < s.start$.
Therefore, in this example, only $s_3$ and $s_4$ require computing $l(r,s)$.
\end{example}

\noindent
\textbf{Algorithm description.}
\Cref{alg:overlap_query_opt} extends \Cref{alg:overlap_query} by incorporating \Cref{theorem:filter2}.
The difference from \Cref{alg:overlap_query} is lines \ref{algo2:cell_threshold_b}--\ref{algo2:cell_threshold_e}.
After computing $\theta_{min}(c_{i,j})$, \Cref{alg:overlap_query_opt} applies case (i) of \Cref{theorem:filter2}, and every interval $s$ in $c_{i,j}$ satisfying this case is included in the join result without computing $l(r,s)$.
Next, after computing $\theta_{max}(c_{i,j})$, \Cref{alg:overlap_query_opt} applies case (ii) of \Cref{theorem:filter2} to filter intervals not included in the result.
Specifically, it computes $l(r,s)$ for each $s \in S_{i,j}^{start}$ that does not satisfy cases (i) and (ii), and once the current $s$ satisfies case (ii), scanning $S_{i,j}^{start}$ is terminated.

\begin{algorithm}[!t]
    \caption{Proposed Algorithm with Optimization}    \label{alg:overlap_query_opt}
    \DontPrintSemicolon
    \KwIn{$\mathcal{G}$, $r$, and $\epsilon$}
    \KwOut{$J_{r} = r \bowtie_\epsilon S$}
    Run lines \ref{algo:overlap_b}--\ref{algo:overlap_e} of \Cref{alg:overlap_query}\;
    \For {$i \gets idx$ to $idx'$}{
        $j \gets$ \textsc{Find-Index}$(A_{i}^{end,max}, \theta_{end}(r))$\;
        \While {$j \leq |A_{i}^{end,max}|$}{
            \If {$A_{i}^{start,min}[j] >\theta_{start}(r)$}{
                $j \gets j + 1$, \textbf{continue}
            }
            \eIf {$i < idx'$}{
                \eIf {$r.end < A_{i}^{end,min}[j]$}{
                    $J_{r} \gets J_{r} \cup \{(r,s)\}$, $\forall s \in S^{st}_{i,j}$
                }{
                    $\theta_{min}(c_{i,j}) \gets A_{i}^{end,min}[j]- \epsilon$\;                        \label{algo2:cell_threshold_b}
                    $\lambda \gets \min\{\max\{r.start, \theta_{min}(c_{i,j})\}, \theta_{start}(r)\}$\;
                    $s \gets$ the first interval in $S^{start}_{i,j}$\;
                    \While{$s.start \le \lambda$}{                                                      \label{algo2:no_comparison_b}
                        $J_{r} \gets J_{r} \cup \{(r,s)\}$\;
                        $s \gets$ next interval in $S^{start}_{i,j}$\;                                  \label{algo2:no_comparison_e}
                    }
                    $\theta_{max}(c_{i,j}) \gets A_{i}^{end,max}[j] - \epsilon$\;                       \label{algo2:calc_cell_threshold_2}
                    $\lambda \gets \min\{\theta_{max}(c_{i,j}), \theta_{start}(r)\}$\;
                    \While {$s.start \le \lambda$}{                                                     \label{algo2:comparison_b}
                        \If {$l(r,s) \geq \epsilon$}{
                            $J_{r} \gets J_{r} \cup \{(r,s)\}$\;
                        }
                        $s \gets$ next interval in $S^{st}_{i,j}$\;                                     \label{algo2:cell_threshold_e}
                    }
                }
            }{
                Run lines \ref{algo2:cell_threshold_b}--\ref{algo2:cell_threshold_e}
            }
            $j \gets j + 1$
        }
    }
\end{algorithm}

\subsection{Batch Processing}  \label{sec:proposal:batch}
\Cref{alg:overlap_query,alg:overlap_query_opt} employ a nested-loop approach and handle each $r \in R$ one by one.
This approach may incur redundant computations.
For example, if $r_1$ and $r_2$ have similar start and end points, computing $r_1 \bowtie_{\epsilon} S$ and $r_2 \bowtie_{\epsilon} S$ tends to share a similar interval access pattern.
Therefore, we consider an approach that groups some intervals $\in R$, say $\{r_a, ..., r_b\}$.
Assume that each $r \in \{r_a, ..., r_b\}$ is similar to each other.
For a given cell $c_{i,j}$, if $l(r_a, s) \geq \epsilon$ or $l(r_a, s) < \epsilon$ for every $s$ in $c_{i,j}$, any $r \in \{r_a, ..., r_b\}$ also has such a case in expectation.
If so, we can further skip interval comparisons.
To maximize the improvement by this approach, it is important to consider how to make groups.

We propose the following approach to make groups.

\begin{definition}  \label{definition:group}
Given an interval $r \in R$, all intervals in the group of $r$ satisfy the following rules:
\begin{itemize}
    \setlength{\leftskip}{-2.0mm}
    \item   $r$ and $r'$ can be mapped into the same cell of $\mathcal{G}$.
    \item   $r.start \leq r'.start \leq r.start + \gamma$, where $\gamma \geq 0$ is a hyper-parameter.
    \item   Let $A_{col}^{max}[j]$ be the minimum value that satisfies $\theta_{start}(r) \leq A_{col}^{max}[j]$.
            Then, $r.end \leq r'.end \leq \min\{r.end + \gamma,  A_{col}^{max}[j] + \epsilon\}$.
\end{itemize}
\end{definition}

\begin{figure}
    \centering
    \includegraphics[width=0.85\linewidth]{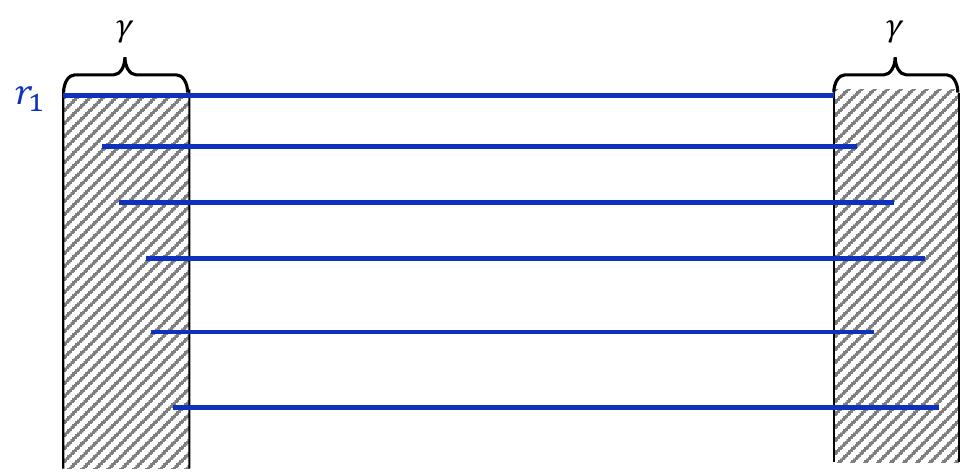}
    \caption{Example of making a group of $r_1 \in R$}
    \label{fig:group}
\end{figure}

\noindent
\Cref{fig:group} illustrates a brief example of making a group of $r_1 \in R$.

To minimize the interval comparisons, we should minimize the number of groups, i.e., similar intervals should be placed in the same group as much as possible.
However, making such groups is NP-hard.

\begin{theorem} \label{theorem:np-hard}
The problem of minimizing the number of groups satisfying Definition \ref{definition:group} is NP-hard.
\end{theorem}

\begin{proof}
We prove this theorem by showing that an instance of this problem includes an NP-hard problem.
Assume that every interval $r \in R$ is mapped to one of the $i$-th column cells, which satisfies the first rule in Definition \ref{definition:group}.
Also, assume that $\epsilon$ is sufficiently large so that $\min\{r.end + \gamma,  A_{col}^{max}[j] + \epsilon\}\min\{r.end + \gamma,  A_{col}^{max}[j] + \epsilon\} =r.end+\gamma$.
Then, let the group of $r$ be $G(r)$, and $r'$ belongs to $G(r)$ iff
\begin{enumerate}
    \item   $r.start \leq r'.start \leq r.start + \gamma$, and
    \item   $r.end \leq r'.end \leq r.end + \gamma$.
\end{enumerate}
Consider that $r$ is a 2-dimensional point of $(r.start,r.end)$.
We can draw a square $\mathcal{S}(r)$ of length $\gamma$ whose left-bottom point is $r$.
If $G(r)$ contains $r'$, $r'$ is inside $\mathcal{S}(r)$ geometrically.

Hence, the problem stated in \Cref{theorem:np-hard} is to partition $R$ into disjoint subsets $G(r_i)$, ..., $G(r_j)$ while minimizing the number of subsets.
From the above geometric representation, this problem is equivalent to the problem of geometric set cover, which is NP-hard \cite{fowler1981optimal}.
Now we see that an instance of our problem contains the NP-hard problem, thereby this theorem holds.
\end{proof}

\noindent
\textbf{Making batches/groups.}
Since \Cref{theorem:np-hard} proves that an optimal set of groups is hard to obtain efficiently in practice, we propose a heuristic approach.
We sort the intervals of $R$ by start point in ascending order.
We make groups (batches) in a greedy manner.

Let $r$ be the first interval in $R$, and we make a group of $r$ by sequentially accessing $R$.
That is, if $r_j$ satisfies Definition \ref{definition:group}, $r_j$ is included in the group of $r$.
(Once the second condition of Definition \ref{definition:group} is not satisfied, we early stop the scan.)
After the scan is terminated, we exclude the group of $r$ from $R$, and iterate the above scan from the first interval in the remaining intervals.
This is repeated until $R$ becomes empty.

Note that making groups does not require query inputs, hence it can be done offline.
If a new interval $r$ is added to $R$, we find a group $G(\cdot)$ that satisfies Definition \ref{definition:group} for $r$, and then add it to the group.
On the other hand, if an existing interval $r'$ is removed from $R$, we simply remove it from the corresponding group.

\vs
\noindent
\textbf{Comparison skipping for $r$ and $s$ in a batch.}
Thanks to the grouping technique, we can further reduce interval comparisons.

\begin{corollary}   \label{corollary:batch}
Given a group $G(r_a)$ of an interval $r_a \in R$, let $r_b$ be the last interval of $G(r_a)$.
If $\theta_{end}(r_b) \leq A_{i}^{end,min}[j]$ and $s.start \leq r_{a}.start$, for every interval $s$ in $c_{i,j}$, we have $(r,s) \in R \bowtie_{\epsilon} S$ for all $r \in G(r_a)$.
\end{corollary}

\begin{proof}
From Definition \ref{definition:group}, it is guaranteed that $r_a.start \leq r_b.start$ and $\theta_{end}(r_a) \leq \theta_{end}(r_b)$.
Hence, from Corollary \ref{corollary:cell_1}, this corollary holds.
\end{proof}

\begin{corollary}   \label{corollary:batch_2}
% Given a group $G(r_a)$ of an interval $r_a \in R$, if $A_{i}^{end,max}[j] < \theta_{end}(r_a) + \gamma$, for every interval $s$ in $c_{i,j}$ and every interval $r \in G(r_a)$, we have $(r,s) \notin R \bowtie_{\epsilon} S$.
Given a group $G(r_a)$ of an interval $r_a \in R$, if $A_{i}^{start,min}[j] > \theta_{start}(r_a) + \gamma$, for every interval $s$ in $c_{i,j}$ and every interval $r \in G(r_a)$, we have $(r,s) \notin R \bowtie_{\epsilon} S$.
\end{corollary}

\begin{proof}
From Corollary \ref{corollary:cell_3}.
\end{proof}

\noindent
\textbf{Algorithm description.}
\Cref{alg:batch} describes the join procedure between $G(r)$ (i.e., an interval group of $r$) and $S$.
It is based on \Cref{alg:overlap_query_opt} and optimized for interval groups.
(This algorithm is iterated for each group to obtain $R \bowtie_{\epsilon}S$.)
The main differences appear in lines \ref{algo:diff_1_b}--\ref{algo:diff_1_e} and lines \ref{algo:diff_2_b}--\ref{algo:diff_2_e}.
The first one uses Corollary \ref{corollary:batch} to include the join result of $G(r) \times S^{end}_{i,j}$ in a batch, and this cell is never accessed again.
The other uses Corollary \ref{corollary:batch_2} to skip the cell for every $r' \in G(r)$, and this cell is never accessed in the future for this group.

\begin{algorithm}[!t]
    \caption{Proposed Algorithm for Batch Intervals} \label{alg:batch}
    \DontPrintSemicolon
    \KwIn{$\mathcal{G}$, $G(r)$, and $\epsilon$}
    \KwOut{$J_{G(r)} = G(r) \bowtie_\epsilon S$}

    $J_{G(r)} \gets \varnothing$\;
    $\theta_{start}(r) \gets r.end - \epsilon$\;
    $\theta_{end}(r) \gets r.start + \epsilon$\;
    $idx \gets$ \textsc{Find-Index}$(A_{col}^{max}, r.start)$\;
    $idx' \gets$ \textsc{Find-Index}$(A_{col}^{max}, \theta_{start}(r))$\;
    $C \gets \varnothing$\;
    \ForEach {$r' \in G(r)$ s.t. $|r'| \geq \epsilon$ in the reverse order}{
        \ForEach {$i = 0$ to $idx - 1$}{
            $j \gets$ \textsc{Find-Index}$(A_{i}^{end,max}, \theta_{end}(r))$\;
            \While{$j \leq |A_{i}^{end,max}|$}{
                \If {$c_{i,j} \notin C$}{
                    \eIf {$\theta_{end}(r') \leq A_{i}^{end,min}[j]$}{              \label{algo:diff_1_b}
                        \ForEach{$s \in S^{end}_{i,j}$ and $r'' \in G(r)$ s.t. $r''.start \leq r'.start$}{         \label{algo:diff_1_b_}
                            $J_{G(r)} \gets J_{G(r)} \cup \{(r'',s)\}$
                        }
                        $C \gets C \cup \{c_{i,j}\}$                                \label{algo:diff_1_e}
                    }{
                        \ForEach{$s \in S^{end}_{i,j}$ s.t. $\theta_{end}(r') \leq s.end$}{
                            $J_{G(r)} \gets J_{G(r)} \cup \{(r',s)\}$
                        }
                    }
                    $j \gets j + 1$
                }
            }
        }
        \ForEach {$i = idx$ to $idx'$}{
            $j \gets$ \textsc{Find-Index}$(A_{i}^{end,max}, \theta_{end}(r))$\;
            \While{$j \leq |A_{i}^{end,max}|$}{
                \If {$c_{i,j} \notin C$}{
                    \eIf {$i < idx'$}{
                        \eIf {$r.end + \gamma < A_{i}^{end,min}[j]$}{
                            Run lines \ref{algo:diff_1_b_}--\ref{algo:diff_1_e}
                        }{
                            Run lines \ref{algo2:cell_threshold_b}--\ref{algo2:cell_threshold_e} of Algorithm \ref{alg:overlap_query_opt}    \label{algo:other}
                        }
                    }{
                        \If {$A_{i}^{start,min}[j] >\theta_{start}(r) + \gamma$}{   \label{algo:diff_2_b}
                            $C \gets C \cup \{c_{i,j}\}$, $j \gets j + 1$\; 
                            \textbf{continue}   \label{algo:diff_2_e}
                        }
                        Run line \ref{algo:other}
                    }
                    $j \gets j + 1$
                }
            }
        }
        % \ForEach {$i = idx$ to $idx'$}{
        %     $j \gets$ \textsc{Find-Index}$(A_{i}^{end,max}, \theta_{end}(r))$\;
        %     \While{$j \leq |A_{i}^{end,max}|$}{
        %         \If {$c_{i,j} \notin C$}{
        %             \If {$A_{i}^{start,min}[j] > \theta_{start}(r)$}{   \label{algo:diff_2_b}
        %                 $C \gets C \cup \{c_{i,j}\}$\;
        %                 $j \gets j + 1$\;
        %                 \textbf{continue}              \label{algo:diff_2_e}
        %             }
        %             \eIf {$i < idx'$}{
        %                 \eIf {$r.end + \gamma < A_{i}^{end,min}[j]$}{
        %                     Run lines \ref{algo:diff_1_b_}--\ref{algo:diff_1_e}
        %                 }{
        %                     Run lines \ref{algo2:cell_threshold_b}--\ref{algo2:cell_threshold_e} of \Cref{alg:overlap_query_opt}    \label{algo:other}
        %                 }
        %             }{
        %                 Run line \ref{algo:other}
        %             }
        %             $j \gets j + 1$
        %         }
        %     }
        % }
    } 
\end{algorithm}

\section{Experiment}    \label{sec:experiment}
This section reports our experimental results.
All experiments were conducted on an Ubuntu 24.04 LTS machine with an Intel Core i9-10980XE@3.0GHz CPU and 128GB of RAM.

\subsection{Setting}    \label{sec:experiment:setting}
\noindent
\textbf{Dataset.}
We used three real-world datasets: BTC\footnote{\url{https://www.kaggle.com/datasets/swaptr/bitcoin-historical-data}}, Books \cite{bouros2021memory}, and Renfe\footnote{\url{https://www.kaggle.com/datasets/thegurusteam/spanish-high-speed-rail-system-ticket-pricing}}, which are used in recent interval works \cite{amagata2024efficient,amagata2024efficient_,amagata2024independent,amagata2024independent_,bouros2025relevance}.
\begin{itemize}
    \setlength{\leftskip}{-2.0mm}
    \item   BTC is a set of historical price intervals of Bitcoin.
            We used low and high prices as the start and end points, respectively.
    \item   Books is a set of the borrowing periods of books in the Aarhus libraries.
    \item   Renfe is Spanish rail trip data.
            We used departure time and arrival time as the start and end points, respectively.
\end{itemize}
\Cref{tab:dataset} shows the statistics of each dataset.
For each dataset, we randomly sampled some intervals, and they were used as $R$, and the remaining intervals were used as $S$.
In our experiments, we varied $|R|/|S|$ in the range of $\{0.25, 0.5, 0.75, 1.0\}$, and, by default, $|R|/|S| = 1$.
Notice that $R \bowtie_{\epsilon} S = S \bowtie_{\epsilon} R$, so the case of $|R|/|S| > 1$ does not occur by exchanging $R$ and $S$.
Furthermore, let $l_{avg}$ be the average length of the intervals in a given dataset.
We used $p \times l_{avg}$ as $\epsilon$, and varied $p$ in the range of $\{0.2, 0.4, 0.6, 0.8\}$.
By default, $p = 0.4$.

\begin{table}[!t]
    \centering
    \caption{Dataset statistics}
    \label{tab:dataset}
    \begin{tabular}{lccc}   \toprule
    \textbf{}           & BTC           & Books         & Renfe    \\ \midrule
    Cardinality         & 3,766,762     & 2,312,602     & 38,753,060        \\
    Domain size         & 6,876,400     & 31,507,200    & 52,163,400        \\
    Minimum length      & 0             & 0             & 1,320             \\
    Maximum length      & 547,077       & 31,406,400    & 44,700            \\
    \bottomrule
\end{tabular}
\end{table}

\vs
\noindent
\textbf{Evaluated algorithm.}
Our experiments evaluated the following algorithms.
\begin{itemize}
    \setlength{\leftskip}{-2.0mm}
    \item   FS \cite{bouros2017forward,bouros2021memory}:
            This is the state-of-the-art interval join algorithm.
            To deal with our problem, i.e., the overlap duration constraint, we extended FS as follows.
            When it identifies that $r$ and $s$ overlap, it computes $l(r,s)$ and adds this pair into the result iff $l(r,s) \geq \epsilon$.
    \item   RD-index \cite{ceccarello2023indexing,ceccarello2025indexing}:
            This is the state-of-the-art algorithm that considers interval duration (not overlap duration).
            Because this algorithm does not consider joins but searches, we run a search query $r \in R$ on $S$, i.e., we ran nested-loop join.
    \item   Rel \cite{bouros2025relevance}:
            This algorithm considers threshold queries, which output every interval for which the overlap duration between a query and the interval is at least the given threshold.
            Similar to RD-index, we ran this search on $S$ by using each $r \in R$ as a query.
            We used an interval tree as its base index.
    \item   Ours\footnote{\url{https://github.com/NaoyaEhara/duration-constrained_interval_joins}}:
            This is our proposed algorithm shown in \Cref{alg:overlap_query_opt}.
            Note that \Cref{alg:batch} is evaluated in \Cref{sec:experiment:batch}.
\end{itemize}
These algorithms were single-threaded, implemented in C++, and compiled by g++ 13.3.0 with -O3 flag.

\subsection{Pre-processing Time}
\Cref{tab:preprocessing} shows the pre-processing (i.e., indexing) time of each algorithm.
Note that each algorithm sorts $R$ and $S$.
Although RD-index, Rel, and ours build a data structure on $S$ only, they sort $R$ by start point to improve the cache hit ratio.

FS needs a shorter pre-processing time because it does not build a complex data structure.
Although RD-index and ours build a similar data structure, ours requires a longer pre-processing time than RD-index.
This is because ours prepares additional sorted arrays shown in \Cref{sec:proposal}.
The pre-processing time of Rel is much slower than that of the others, because it builds a hierarchical structure and requires recursive partitioning.
This processing is more complex than building a grid-based structure.

\begin{table}[!t]
    \centering
    \caption{Pre-processing time [sec]}
    \label{tab:preprocessing}
    \begin{tabular}{lccc}   \toprule
        \textbf{}   & BTC   & Books & Renfe \\ \midrule
        FS          & 0.24  & 0.11  & 2.42  \\
        RD-index    & 0.48  & 0.24  & 3.95  \\
        Rel         & 3.49  & 0.71  & 28.97 \\
        Ours        & 0.57  & 0.27  & 4.37  \\
    \bottomrule
    \end{tabular}
\end{table}

\subsection{Memory Usage}
\Cref{tab:memory} exhibits the memory usage of each algorithm.
Note that FS does not build any data structure offline, so its memory usage is omitted.
Rel requires more memory than the others, as it employs a hierarchical structure with a non-contiguous memory layout.
Ours requires almost twice the memory usage of RD-index.
This is because ours maintains intervals in two arrays, sorted based on start and end points.
Nevertheless, the memory usage of our structure is reasonable for about 40 million intervals (i.e., Renfe).

\begin{table}
    \centering
    \caption{Memory usage [MB]}
    \label{tab:memory}
    \begin{tabular}{lccc} \toprule
        \textbf{}   & BTC       & Books & Renfe     \\ \midrule
        % FS & 0 & 0 & 0 \\
        RD-index    & 32.06     & 20.53 & 327.66    \\
        Rel         & 198.28    & 94.50 & 929.87    \\
        Ours        & 63.84     & 39.66 & 653.62    \\
        \bottomrule
    \end{tabular}
\end{table}

\subsection{Join Time}
\Cref{fig:size_ratio} shows the results of experiments varying $|R|/|S|$.
As $|R|/|S|$ increases, the join result size increases, thereby all algorithms need longer join time.
Our algorithm outperforms all competitors and usually requires 0.5x the time of the best baseline algorithm.
This result demonstrates the efficacy of our pruning rules based on \Cref{theorem:filter,theorem:filter2}.

Although RD-index and ours share a similar data structure, the performance difference is substantial.
This difference highlights the effectiveness of exploiting the grid structure for \Cref{theorem:filter,theorem:filter2}.
Ours can filter $(r,s)$ quickly if $0 < l(r,s) < \epsilon$, whereas FS has to evaluate all interval pairs $\in R \bowtie S$.
Therefore, ours is always faster than FS.
Rel also has an ability that filters $(r,s)$ if $0 < l(r,s) < \epsilon$.
The performance difference between ours and Rel suggests that our pruning is more efficient and tighter.

\begin{figure*}[!t]
    \begin{center}
        \subfigure[BTC]{%
            \includegraphics[width=0.32\linewidth]{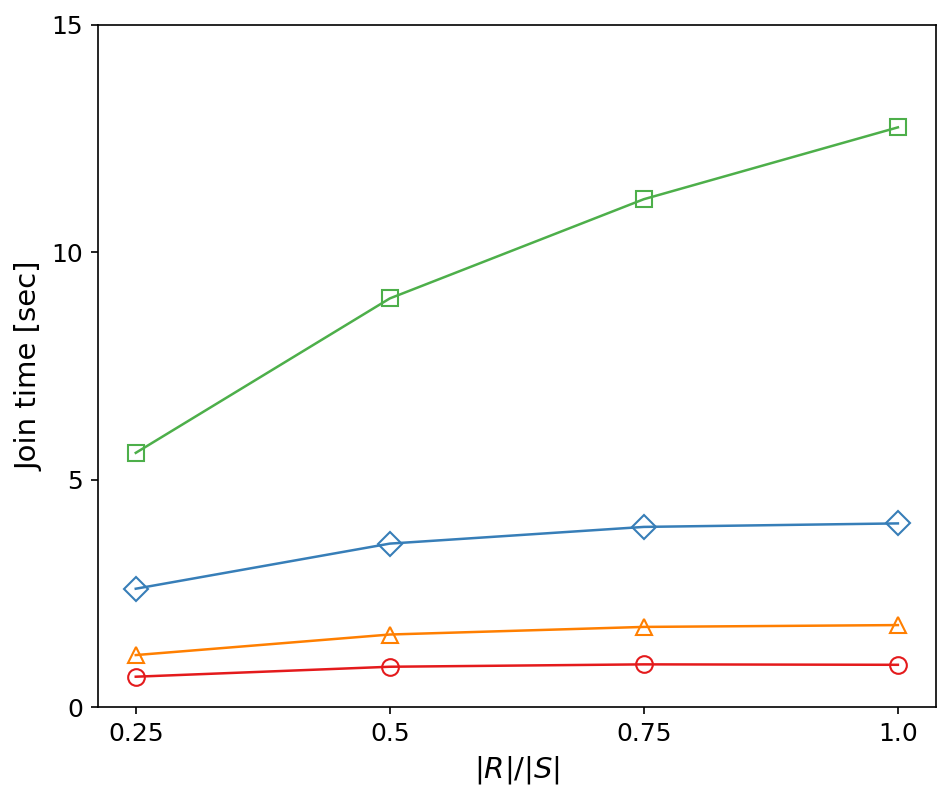}    \label{fig:size_btc}}
        \subfigure[Books]{%
            \includegraphics[width=0.32\linewidth]{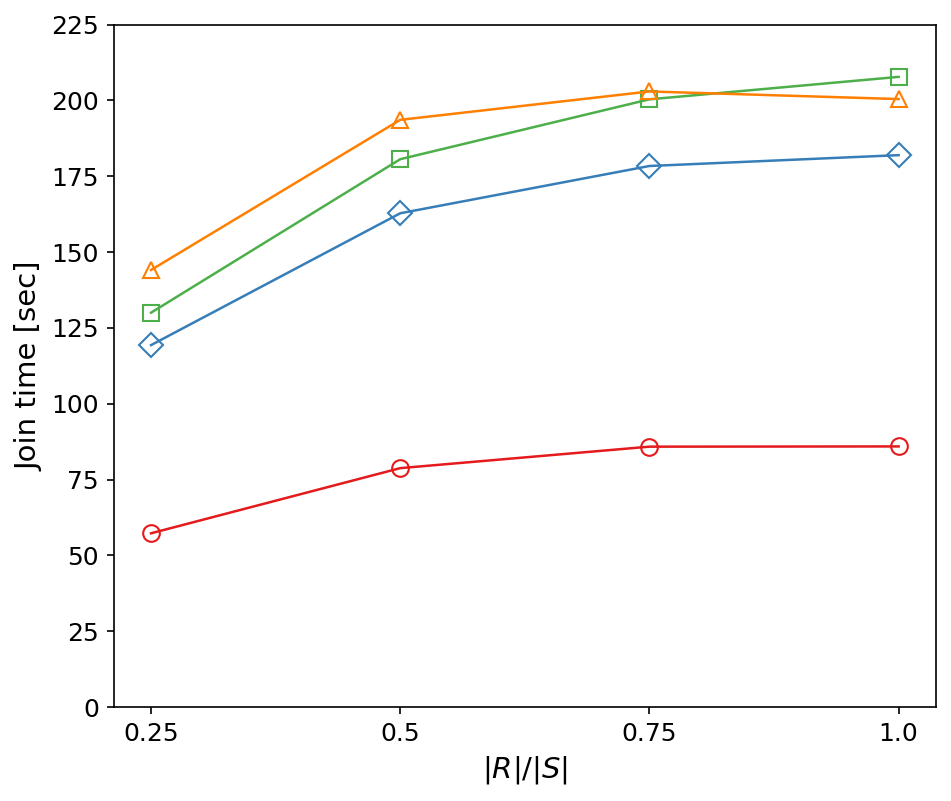}  \label{fig:size_books}}
        \subfigure[Renfe]{%
            \includegraphics[width=0.32\linewidth]{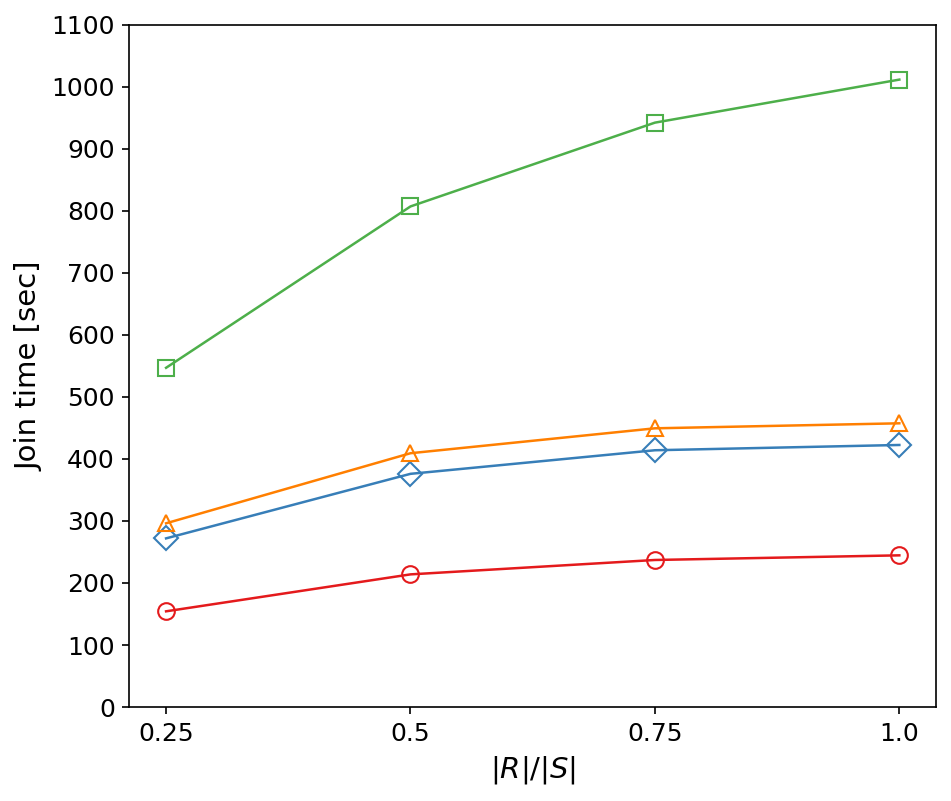}  \label{fig:size_renfe}}
        \caption{Impact of $|R|/|S|$: ``\textcolor{red}{$\circ$}'' shows \textcolor{red}{ours}, ``\textcolor{blue}{$\diamond$}'' shows \textcolor{blue}{FS}, ``\textcolor{green}{$\Box$}'' shows \textcolor{green}{RD-index}, and ``\textcolor{orange}{$\triangle$}'' shows \textcolor{orange}{Rel}.}
        \label{fig:size_ratio}
    \end{center}
\end{figure*}
\begin{figure*}[!t]
    \begin{center}
        \subfigure[BTC]{%
            \includegraphics[width=0.32\linewidth]{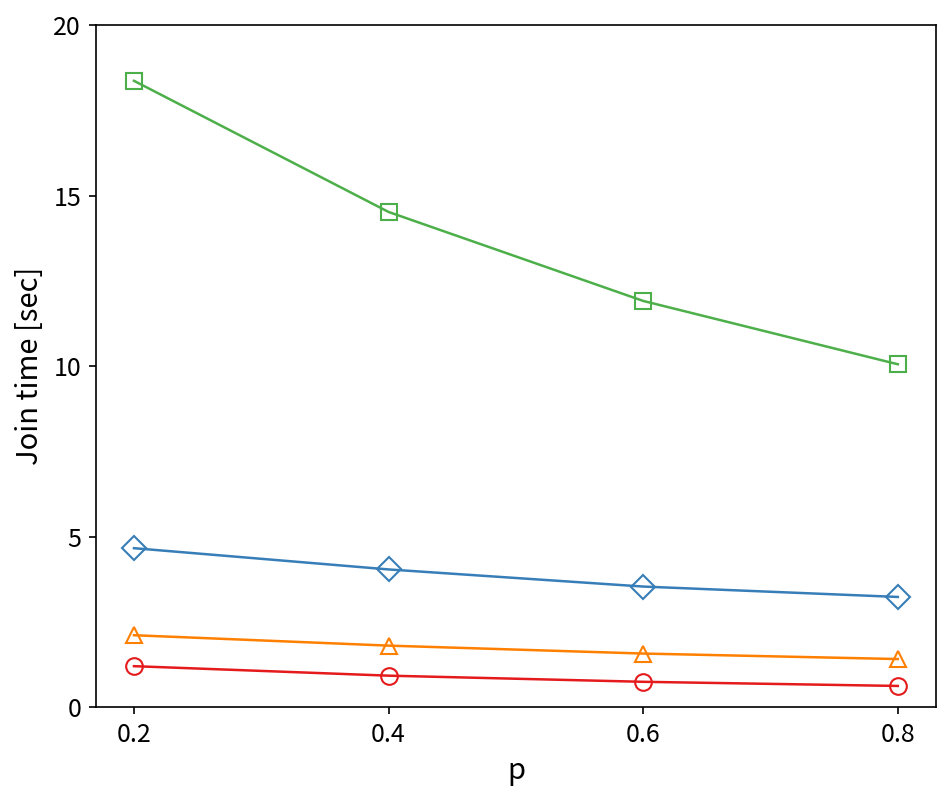}    \label{fig:eps_btc}}
        \subfigure[Books]{%
            \includegraphics[width=0.32\linewidth]{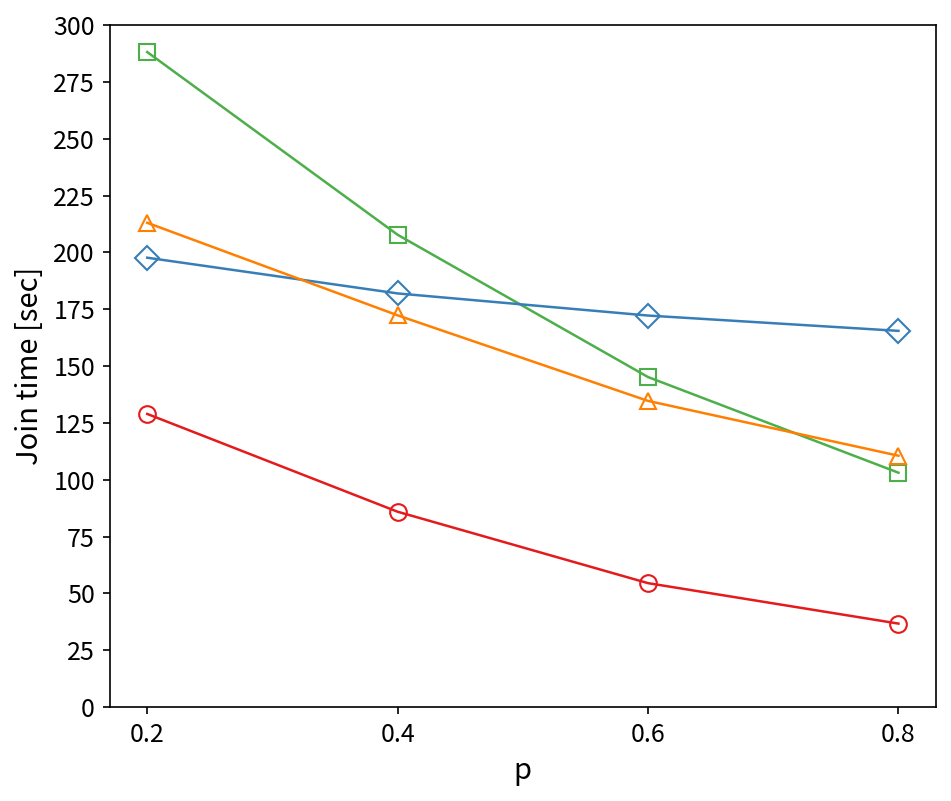}  \label{fig:eps_books}}
        \subfigure[Renfe]{%
            \includegraphics[width=0.32\linewidth]{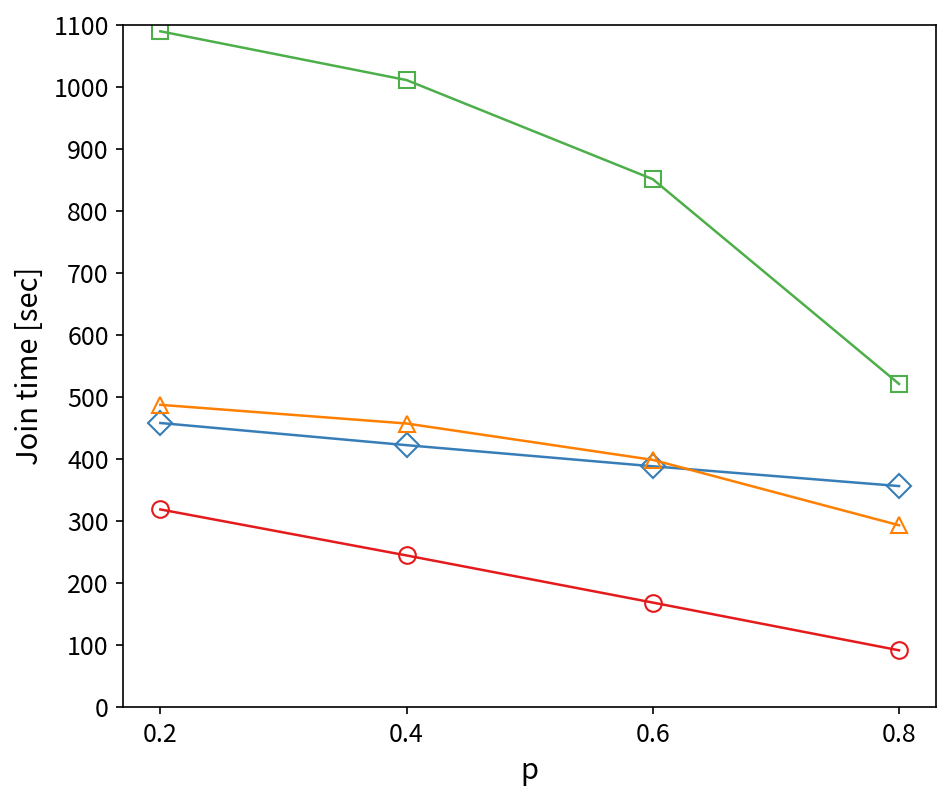}  \label{fig:eps_renfe}}
        \caption{Impact of $\epsilon$ ($p$): ``\textcolor{red}{$\circ$}'' shows \textcolor{red}{ours}, ``\textcolor{blue}{$\diamond$}'' shows \textcolor{blue}{FS}, ``\textcolor{green}{$\Box$}'' shows \textcolor{green}{RD-index}, and ``\textcolor{orange}{$\triangle$}'' shows \textcolor{orange}{Rel}.}
        \label{fig:epsilon}
    \end{center}
\end{figure*}

\subsection{Impact of $\epsilon$}
We next investigate the impact of $\epsilon$ by varying $p$, see \Cref{sec:experiment:setting}.
\Cref{fig:epsilon} shows the result.
As $p$, i.e., $\epsilon$, increases, $|R \bowtie_{\epsilon} S|$ decreases.
That is, many $(r,s)$ become unqualified, thereby pruning efficiency improves.
Even when $p$ is small or large, our algorithm still keeps outperforming the baseline algorithms.

\subsection{Ablation Study}
Here, we investigate the efficiency of the optimization proposed in \Cref{sec:proposal:opt}.
We thus compared the performance of \Cref{alg:overlap_query,alg:overlap_query_opt}.
\Cref{tab:ablation} shows the comparison result.

We observe that ours, i.e., \Cref{alg:overlap_query_opt}, consistently outperforms ours without opt, i.e., \Cref{alg:overlap_query}.
This result demonstrates that our proposed optimization functions well in practice and reduces unnecessary comparisons between intervals.

\begin{table}[!t]
    \centering
    \caption{Ablation study (join time [sec])}
    \label{tab:ablation}
    \begin{tabular}{lccc} \toprule
        \textbf{}                   & BTC  & Books  & Renfe     \\ \midrule
        Ours                        & 0.92 & 85.92  & 244.53    \\
        Ours without optimization   & 1.61 & 139.02 & 393.36    \\
        \bottomrule
    \end{tabular}
\end{table}

\subsection{Impact of Batch Processing} \label{sec:experiment:batch}
\Cref{tab:batch_preprocessing} shows the time to form groups/batches\footnote{We empirically tuned $\gamma$ to optimize the join time, and set $\gamma$ was $2 \times l_{avg}$, $0.12 \times l_{avg}$, and $0.0004 \times l_{avg}$ on BTC, Books, and Renfe, respectively.}.
Compared with \Cref{tab:preprocessing}, we find that grouping is not a bottleneck.
Recall that it is pre-processing, so the time to form groups is reasonable.

\begin{table}[!t]
    \centering
    \caption{Grouping time [sec]}
    \label{tab:batch_preprocessing}
    \begin{tabular}{lccc} \toprule
        \textbf{}       & BTC   & Books & Renfe \\ \midrule
        Ours (batch)    & 0.46  & 0.14  & 2.60  \\
        \bottomrule
    \end{tabular}
\end{table}

\Cref{tab:batch} shows the efficacy of \Cref{alg:batch}.
We see that our batch processing is useful when intervals are dense (i.e., Books and Renfe), whereas it is not when intervals are sparse (i.e., BTC).
If batch filtering (i.e., the main differences to \Cref{alg:overlap_query_opt}) fails, it becomes an additional cost.
The sparse case usually has such failures, so the batch-based join is slower than the nested-loop-based join on BTC.
On the other hand, the dense case usually succeeds the batch filtering, thereby reducing the total time from the nested-loop-based join.

\begin{table}[!t]
    \centering
    \caption{Impact of batch processing (join time [sec])}
    \label{tab:batch}
    \begin{tabular}{lccc} \toprule
        \textbf{}                   & BTC  & Books & Renfe  \\ \midrule
        Ours                        & 0.92 & 85.92 & 244.53 \\
        Ours with batch processing  & 1.45 & 50.20 & 151.15 \\
        \bottomrule
    \end{tabular}
\end{table}

\section{Related Work}  \label{sec:related-work}
This section reviews existing works w.r.t. interval joins and range search.

\subsection{Interval Join}
The problem of interval join is defined in Definition \ref{def:join}.
Many works addressed this problem \cite{gunadhi1991query,bouros2020band,bouros2017forward,bouros2018interval,bouros2021memory,cafagna2017disjoint,dignos2014overlap,piatov2016interval,piatov2021cache,dignos2022leveraging}.
We introduce some representative algorithms.

OIP \cite{dignos2014overlap} divides the domain of a given interval set into disjoint partitions.
Each interval is assigned to the minimum number of partitions that cover this interval.
An interval join query is processed based on partitions of $R$ and $S$.
That is, mini-joins are performed between overlapping partitions.
DIP \cite{cafagna2017disjoint} improves this approach by arranging each partition so that its intervals are disjoint.
This disjoint property reduces the number of interval comparisons.

EBI and LEBI \cite{piatov2016interval,piatov2021cache} are join algorithms based on the plane-sweep method.
They are designed to be cache-efficient and parallel-friendly.
FS \cite{bouros2017forward,bouros2021memory} is the state-of-the-art interval join algorithm.
It is also based on the plane-sweep method and improves EBI and LEBI.
FS alleviates unnecessary interval comparisons when the sweep-line stops at a start or an end point, compared with EBI and LEBI.
FS includes optimizations, such as bucketing, grouping, loop unrolling, and a decomposed data layout, which also support multi-threading. 
Note that their performances tend to degrade when dealing with long intervals, e.g., see \Cref{fig:size_books}.

OMJ \cite{dignos2022leveraging} decomposes an interval join into two range joins.
Its approach is compatible with existing database management systems (e.g., PostgreSQL), because it employs B$^+$-tree to run range joins.
However, as the experimental results in \cite{dignos2022leveraging} confirm, it is (slightly) outperformed by FS.

The above algorithms do not consider the overlap duration, so they need to employ one of the two approaches introduced in \Cref{sec:challenge} to handle our problem.
In \Cref{sec:experiment}, we empirically show that this approach is slower than our algorithm.

\subsection{Range Search}
A range search on interval data $S$ retrieves a subset $S' \subseteq S$ such that $s \in S'$ satisfies $l(r,s) > 0$ for a given query interval $r$.
(Similar to our join algorithm, interval join can be conducted by repeating range search of $r$ on $S$ for every $r \in R$.)
A classic yet still state-of-the-art range search structure is the interval tree \cite{edelsbrunner1980dynamic}.
It can return $S'$ in $O(\log n + |S'|)$ time.
Timeline index \cite{kaufmann2013timeline} is implemented in SAP HANA \cite{farber2012sap}.
This index is essentially a list, where the start and end points of all intervals are sorted.
HINT \cite{christodoulou2023hint,christodoulou2022hint} is a hierarchical one-dimensional grid.
Its design objective is similar to ours: HINT can identify intervals in the search result in a comparison-free manner.
\cite{wang2025tide} considers the range search problem in the external memory setting.
These works do not consider any duration constraint, so they are not appropriate for our problem.

Period index \cite{behrend2019period} considers range-duration queries.
The concept of range-duration queries is similar to our duration constraint, but it is different.
A range-duration query $r$ returns every interval $s \in S$ such that $l(r,s) > 0$ and $d_{min} \leq |s| \leq d_{max}$, where $d_{min}$ and $d_{max}$ are the duration constraints.
That is, this query considers not the overlap duration but the interval duration.
The period index is also a hierarchical one-dimensional grid.
RD-index \cite{ceccarello2023indexing,ceccarello2025indexing} also considers range-duration queries and improves the period index.
We extended the RD-index to deal with the overlap duration constraint.
\Cref{sec:experiment} confirms that our algorithm significantly outperforms this version of the RD-index.
In \cite{bouros2025relevance}, the authors define relevance queries: They measure the similarity between two intervals based on the Jaccard similarity, overlap ratio, and overlap duration.
They proposed lower- and upper-bounding techniques for these measures to filter unnecessary similarity computations.
In \Cref{sec:experiment}, we compared our algorithm with this approach, and the result demonstrates that our algorithm is faster in our problem.

To reduce the output size, \cite{amagata2024independent,amagata2024independent_,amagata2026firas} consider the independent range sampling problem on interval data, and this problem returns $t$ random samples of the range search result.
Some works \cite{amagata2025top,amagata2024efficient,amagata2024efficient_,xu2017efficiently,lee2026efficient} assume that each interval has a weight, and they consider returning $k$ intervals overlapping a given query interval with the best weight.
On the other hand, \cite{rauch2025fast} considers keywords as a constraint, i.e., it considers returning intervals that overlap a given query interval and have the specified keywords.
The settings of these works are totally different from ours.

\section{Conclusion}    \label{sec:conclusion}
This work addressed the problem of duration-constrained interval join for the first time.
Applying existing join and search techniques to this problem is inefficient, so we proposed a new algorithm that filters unnecessary computations in a batch.
To improve this algorithm, we proposed two optimization techniques.
We conducted extensive experiments on real-world interval data, and the results demonstrate that our algorithm outperforms existing techniques.

\vs
\noindent
\textbf{Remark after acceptance.}
Although the proceedings version claims that our problem has not been studied yet, \cite{hu2022computing} already considers it.
We hence remove this claim in this version.

\section*{Acknowledgment}
This research is partially supported by JSPS KAKENHI Grant Number 24K14961.

\bibliographystyle{IEEEtran}
\bibliography{acmart}

\end{document}